\documentclass[letterpaper, 10 pt, conference]{ieeeconf}  

\IEEEoverridecommandlockouts                              
\usepackage[utf8]{inputenc}
\usepackage[english]{babel}

\usepackage{amssymb}
\usepackage{amsmath}

\usepackage{amsthm}
\usepackage{mathtools}
\usepackage{graphicx}
\usepackage{wrapfig}
\usepackage{subcaption}
\let\labelindent\relax

\usepackage{enumitem}
\usepackage{caption}
\usepackage{float}
\usepackage{multicol}
\usepackage{hyperref}
\usepackage{xcolor}
\usepackage{eurosym}
\usepackage{dsfont}
\usepackage{bm}
\usepackage[ruled, vlined, linesnumbered]{algorithm2e}

\usepackage[left=1.91cm, right=1.91cm, bottom=2.1cm, top = 1.91cm]{geometry}

 \newcommand{\tcr}{}
\newtheorem{prop}{Proposition}

\newtheorem{definition}{Definition}

\newtheorem{rmq}{Remark}

\DeclareMathOperator*{\argmin}{arg\,min}

\title{
Learning Pure Nash Equilibrium in Smart Charging Games
}

\author{B. Sohet and Y. Hayel and O. Beaude and A. Jeandin
\thanks{B. Sohet and O. Beaude and A. Jeandin are with EDF Labs, France
        {\tt\small benoit.sohet@edf.fr}}%
\thanks{Y. Hayel is with LIA/CERI, Avignon University, France
        {\tt\small yezekael.hayel@univ-avignon.fr}}%
}

\begin{document}

\maketitle
\thispagestyle{empty}
\pagestyle{empty}

\begin{abstract}
Reinforcement Learning Algorithms (RLA) are useful machine learning tools to understand how decision makers react to signals.
It is known that RLA converge towards the pure Nash Equilibria (NE) of finite congestion games and more generally, finite potential games.
For finite congestion games, only separable cost functions are considered.
However, non-separable costs, which depend on the choices of all players instead of only those choosing the same resource, may be relevant in some circumstances, like in smart charging games.
In this paper, finite congestion games with non-separable costs are shown to have an ordinal potential function, leading to the existence of an action-dependent continuous potential function.
The convergence of a synchronous RLA towards the pure NE is then extended to this more general class of congestion games.
Finally, a smart charging game is designed for illustrating convergence of such learning algorithms.
\end{abstract}
\vspace{1mm}


\vspace{-1mm}
\section{Introduction}

Understanding how decision makers react to signals is the basement of complex systems analysis. In such systems where decision makers are selfish and the outcome of each individual depends on the decision of others (typically a game theoretical framework), it is important to compute the equilibrium situation (i.e., Nash Equilibrium, NE) but also to understand how each player will adapt her strategy through time and therefore learn to play her Nash strategy. Learning algorithms in game theoretical settings have been known for several years \cite{Fudenberg98}. Many techniques are based on the well known \textit{best response} principle like fictitious play \cite{Brown51}, which assumes that players choose a best reply to the observed empirical distribution of past actions of the other players. The main drawback of such techniques is the computation of the best action, which can be computationally complex and also needs specific information (utility functions, set of actions, ...). Other types of algorithms, based on reinforcement mechanisms, have also been employed in game theoretical problems. These types of decentralized learning techniques, based on experiments as in repeated games, proved to be efficient in particular games such as congestion games and more generally potential games, which have convergence properties, as illustrated recently in machine learning community with different feedback information~\cite{Cohen17}.

Congestion games are a specific type of non-cooperative games in which the cost of a commodity depends on the number of players choosing it~\cite{Nisan07}.
Some costs functions are said non-separable, like the energy cost for Electic Vehicle (EV) users, when it relies on a smart dynamic pricing. This means that the cost perceived by an individual choosing one commodity or route, depends on the choices of all players and not only on the one choosing the same commodity. Non-separability property is not present in standard congestion games where cost is only due to congestion. Few papers deal with congestion games with non-separable cost functions because of the non-existence of a Beckmann function \cite{Beckmann56}, except in particular cases~\cite{Sohet20}. Most of them focus on the Price of Anarchy, which is a measure of the performance induced by a     decentralized system compared to centralized optimization. Learning techniques in congestion games are not very studied because it mainly considers non-atomic games, i.e. a population of players like an infinite number of decision makers. In \cite{Barth09}, the authors adapt the reinforcement learning algorithm (RLA) described in \cite{SASTRY94} to atomic congestion games, i.e. with a finite number of decision makers.
In their adaptation of the RLA, costs functions are separable. In this paper, this assumption is relaxed and non-separable costs are considered, such as energy costs in smart charging games.
\tcr{In~\cite{Bournez09}, it is shown that an asynchronous RLA converges for ordinal potential games.}

Smart charging EV is a particularly interesting environment to implement online learning algorithms to deal with different degrees
of uncertainty and randomness of future knowledge. Most papers related to smart charging consider machine learning techniques from a centralized point of view \cite{Tang16}. Deep learning techniques are suggested in \cite{Qian20} to study EV charging navigation and to minimize the total travel time and the charging cost at charging stations, with also a centralized point of view. In~\cite{Wang17}, the NE of the constrained energy trading game among players with incomplete information is found using a RLA. This game however does not consider a smart charging context. In our paper, we adapt RLA techniques to a finite smart charging game and show convergence to a pure NE. The contributions of this paper are the following:
 \vspace{-1mm}
\begin{itemize}
    \item \tcr{Proof of potentiality of} an atomic congestion game with non-separable cost functions;
    \item Proof of RLA convergence for this game;
    \item Application to a smart charging game and numerical example.
\end{itemize}

The paper is organized as follows: finite congestion games with non-separable costs are introduced in Section \ref{model} and an action-dependent continuous potential is defined. This leads to convergence results of RLA in such games in Section \ref{RLA}. This methodology is applied to a smart charging game described in Section \ref{Smart} in which numerical illustrations show its convergence. Finally Section \ref{conc} concludes the paper.

\section{Potential functions in finite congestion games with non-separable costs}
\label{model}
In the atomic game $\mathcal{G}$ considered in this work, there are $N$ players.
Each player $i\in \mathcal{N}\triangleq\{1,\dots,N\}$ chooses her resource $r_i$ among the same set $\mathcal{R}$.
This set $\mathcal{R}$ of pure strategies is made of $M$ resources, so that this game is finite.
Note that the following study applies directly to the case where each player $i$ chooses a set $\mathcal{R}_i \subset \mathcal{R}$ of resources.
In a classical congestion game~\cite{Rosenthal73}, a player $i$ choosing resource $r_i$ would have to pay for a cost $c_{r_i}\left(n_{r_i}(\bm{r})\right)$ which depends on the number of players having chosen the same resource, $n_{r_i}(\bm{r})$, defined as:
\vspace{-1mm}
\begin{equation}
    \forall a\in\mathcal{R}\,,\quad n_a\left(\bm{r}\right) \triangleq \#\left\{i\in\mathcal{N} ~|~ r_i=a\right\}\,,
\end{equation}
with $\bm{r} \triangleq \left(r_1, \dots, r_N\right) =\left(r_i, \bm{r}_{-i}\right)$ the vector of actions and $\bm{r}_{-i}$ the vector of actions of the players other than $i$. 
Here, the costs are non-separable~\cite{Chau03}, meaning that the cost $c_a$ of a resource $a$ does not only depend on $n_a$, but also on the number of players $n_b$ choosing other resources $b\in \mathcal{R}\smallsetminus\{a\}$. 
More precisely, \tcr{we can prove the existence of a potential function for specific non-separable costs:}

\vspace{-1mm}
\begin{definition}
Linearly non-separable congestion costs $c_a$, for all resources $a\in\mathcal{R}$ and vectors of actions $\bm{r}\in\mathcal{R}^N$, are defined as:
\vspace{-3mm}
\begin{equation}
    c_a(\bm{r}) \triangleq \alpha_a \lambda\big(L\left(\bm{r}\right)\big)\,,
    \label{eq:nonsep_costs}
\end{equation}
\vspace{-4mm}
\begin{equation*}
\text{with:}\quad
    \begin{cases}
    \alpha_a \geq 0\,, \text{ constant,}\\
    \lambda : \mathbb{R}_+\rightarrow \mathbb{R}_+ \,, \text{ an increasing function,}\\
    L\left(\bm{r}\right) \triangleq \tcr{\sum_{i=1}^N \alpha_{r_i} =} \sum_{b\in \mathcal{R}} \alpha_b n_b(\bm{r}).
    \end{cases}
\end{equation*}
\label{def:nonsep}
\vspace{-3mm}
\end{definition}
Note that \tcr{we chose the term ``linearly" because of} function $L(\cdot)$, which is a linear combination of the number of players $n_a(\cdot)$ choosing each resource $a$, while $\lambda$ may be non-linear.
\vspace{-1mm}
\begin{rmq}
Here, the game is supposed symmetric (between the players), meaning that the cost of a player $i$ only depends on her choice, and not on the player herself.
However, note that the following study can be extended to a non-symmetric case where player $i$ gets the cost $c_{i,r_i} \triangleq \alpha_{i,r_i}\lambda\left(\sum_{j=1}^N \alpha_{j,r_j}\right)$, which is the same as~\eqref{eq:nonsep_costs} when $\alpha_{i,a} = \alpha_a$ for all $i,a$.
\label{rmq:hetero}
\vspace{-1mm}
\end{rmq}

In this particular context of atomic congestion game with non-separable cost functions, we are able to find potential properties for the game, which is a powerful tool for the study of NE in pure strategy and convergence of learning procedures~\cite{Bournez09}.

\subsection{Potential function of pure strategies}
\label{sec:pure}

Following Definition~\ref{def:nonsep} of linearly non-separable congestion costs $\left(c_a\right)_{a\in\mathcal{R}}$, it is possible to extend the ordinal potential property of separable congestion games~\cite{Milchtaich96} to the non-separable game $\mathcal{G}\triangleq\left(\mathcal{N}, \mathcal{R}^N, \left(c_a\right)\right)$.

\vspace{-1mm}
\begin{definition}
An ordinal potential for game $\mathcal{G}$ is a function $P:\mathcal{R}^N\rightarrow \mathbb{R}$ verifying $\forall i\in\mathcal{N}, \forall \bm{r}_{-i}\in\mathcal{R}^{N-1}, \forall a,b \in \mathcal{R}$,
\begin{equation}
    c_{a}\left(a, \bm{r}_{-i}\right) < c_{b}\left(b, \bm{r}_{-i}\right)
    \Leftrightarrow P\left(a,\bm{r}_{-i}\right) < P\left(b,\bm{r}_{-i}\right)\,.
\end{equation}
\end{definition}

This definition follows the idea that an ordinal potential function follows the sign of the difference of cost for any player that changes her action unilaterally. Even if our game $\mathcal{G}$ is not a standard congestion game due to the non-separability of costs functions, it is possible to show the existence of an ordinal potential function. 

\vspace{-1mm}
\begin{prop}
The finite non-cooperative congestion game $\mathcal{G}$ with non-separable costs defined in~\eqref{eq:nonsep_costs} has the following ordinal potential function:
\vspace{-1mm}
\begin{equation}
    \forall \bm{r}\in\mathcal{R}^N\,, \quad P(\bm{r}) \triangleq \lambda\left(L\left(\bm{r}\right)\right)\,.
\end{equation}
\label{prop:pot}
\vspace{-6mm}
\end{prop}
\begin{proof}
Let $i\in\mathcal{N}$ be any player and $a,b\in\mathcal{R},\bm{r}_{-i}\in\mathcal{R}^{N-1}$ any actions.

Firstly, note that $L\left(a,\bm{r}_{-i}\right)=L\left(b,\bm{r}_{-i}\right) + \alpha_a - \alpha_b$, by definition.
Then, as function $\lambda$ is increasing:
\begin{equation*}
\begin{aligned}
 \lambda\big(L\left(a,\bm{r}_{-i}\right)\big) &< \lambda\big(L\left(b,\bm{r}_{-i}\right)\big)\hspace{-4mm}
&&\iff
\alpha_a - \alpha_b<0 \,,\\
    \text{i.e.}\qquad P\left(a,\bm{r}_{-i}\right) &< P\left(b,\bm{r}_{-i}\right)
    &&\iff
    \alpha_a < \alpha_b\,,
    \end{aligned}
\end{equation*}
by definition of $P = \lambda \circ L$.

Secondly, function $C_{\bm{r}} : \alpha \mapsto \alpha \times \lambda\big(\alpha+\sum_{j\neq i}\alpha_{r_j}\big)$ is increasing on $\mathbb{R}_+$ as a product of positive increasing functions, meaning that:
\begin{equation*}
\begin{aligned}
 \alpha_a \lambda\Big(\alpha_a+\sum_{j\neq i}\alpha_{r_j}\Big) &< \alpha_b \lambda\Big(\alpha_b+\sum_{j\neq i}\alpha_{r_j}\Big) \hspace{-4mm}
&&\iff
\alpha_a < \alpha_b \,,\\
\text{i.e.} \qquad 
c_{a}\left(a, \bm{r}_{-i}\right) &< c_{b}\left(b, \bm{r}_{-i}\right)
    &&\iff \alpha_a < \alpha_b\,.
    \end{aligned}
\end{equation*}
\vspace{-2mm}
\vspace{-1mm}
\end{proof}

\noindent
The existence of a potential implies the existence of a pure NE in such non-cooperative games, which is defined here:


\vspace{-1mm}
\begin{definition}
A pure strategy Nash Equilibrium (NE) of game $\mathcal{G}$ is a vector of actions $\bm{r}^*\in\mathcal{R}^N$ which verifies:
\vspace{-1mm}
\begin{equation}
    \forall i\in\mathcal{N}\,,
    \forall a\in\mathcal{R}\,, \quad
    c_{r^*_i}\left(r^*_i, \bm{r}^*_{-i}  \right) \leq c_a \left(a, \bm{r}^*_{-i}  \right)\,.
\end{equation}
\end{definition}
\noindent
In other words, a NE is a strategy vector such that no player can reduce her cost by changing her strategy unilaterally. Note that the existence of pure NE is not a standard result, but it is true for games with an ordinal potential function. Indeed, as the sets of actions are compact, the minimum of the potential exists and corresponds to a pure NE of the game~\cite{Monderer96}. In this particular non-separable atomic congestion game, pure NE can be fully characterized. Let us define the set of resources $\mathcal{R}^+\triangleq\{a\in\mathcal{R}~|~\alpha_a > \min_{s\in\mathcal{R}}\left(\alpha_s\right)\}$.
\begin{prop}
The NE of $\mathcal{G}$ are the $\bm{r}^*\in\mathcal{R}^N$ such that:
\vspace{-1mm}
\begin{equation}
    \forall a\in\mathcal{R}^+\,,\quad n_a\left(\bm{r}^*\right) = 0\,.
    \label{eq:NE}
\end{equation}
\end{prop}
\begin{proof}
First, let $\bm{r}^*\in\mathcal{R}^N$ verify~\eqref{eq:NE}.
Then, for any player $i\in\mathcal{N}$ and any resource $a\in\mathcal{R}$,
$\alpha_{r^*_i}=\min_{s\in\mathcal{R}}\left(\alpha_s\right)$, by definition of $\bm{r}^*$.
Thus, $\alpha_{r^*_i} \leq \alpha_a$ and $C_{\bm{r}^*}\left(\alpha_{r^*_i}\right)\leq C_{\bm{r}^*}\left(\alpha_a\right)$, with increasing function $C_{\bm{r}^*} : \alpha \mapsto \alpha \times \lambda\left(\alpha+\sum_{j\neq i}\alpha_{r^*_j}\right)$ defined in proof of Prop.~\ref{prop:pot}.
This is equivalent to $c_{r^*_i}\big(r^*_i, \bm{r}^*_{-i}\big)\leq c_a\left(a,\bm{r}^*_{-i}\right)$.

Secondly, let $\bm{r}^*\in\mathcal{R}^N$ be such that there exists $a\in\mathcal{R}^+$ with $n_a(\bm{r}^*)\geq 1$.
Let $i$ be one of the players having chosen resource $a$.
Let $b\in\mathcal{R}$ be such that $\alpha_b = \min_{s\in\mathcal{R}}(\alpha_s)$.
Then, $\alpha_a > \alpha_b$ by definition of $\mathcal{R}^+$, and $c_{r^*_i}\left(r^*_i, \bm{r}^*_{-i}\right) > c_b\left(b,\bm{r}^*_{-i}\right)$ using again function $C_{\bm{r}^*}$.
\vspace{-2mm}
\vspace{1mm}
\end{proof}
\noindent
This proposition shows that in games following Definition~\ref{def:nonsep}, pure NE correspond to situations where all players choose the resource $a$ with the lowest coefficient $\alpha_a$.

This parameter may not be known by players in advance, hence the need of learning algorithms by players in order to optimally adapt their actions.
Such RLA are fully decentralized and are based on updates, for each player, of probability distributions over pure strategies (called mixed strategies).
It is shown in~\cite{SASTRY94} that there is a link between the potential function in pure strategies and the one in mixed strategies for common payoff games.
This result is extended to more general games and is fundamental in order to prove the convergence of RLA to pure NE.

\subsection{Action-dependent continuous potential}
Let $\pi_{i,a}$ denote the probability with which player $i$ chooses pure strategy $a \in \mathcal{R}$, $\bm{\pi}_i \in \Delta_i$ the mixed strategy vector of player $i$ in a simplex $\Delta_i$ of $\mathbb{R}^M$ and $\bm{\pi}\in\Delta = \prod_i \Delta_i$ the mixed strategies of all players.
The mixed strategy notation of player $i$ playing pure strategy $a$ is $\bm{\pi}_i=\bm{e}_a$, with $\bm{e}_a$ the null vector except for the $a$-th component, equal to 1.
The expected cost $\overline{c_{i,a}}$ for player $i$ playing pure strategy $a$ is:
\begin{equation}
\label{eq:mixed}
\begin{aligned}
     \overline{c_{i,a}}(\bm{\pi}) \triangleq \mathbb{E}_{\bm{\pi}}\big[c_a | \pi_{i}=\bm{e}_a\big]
     = \sum_{\bm{r}_{-i}} \Big( c_a(a, \bm{r}_{-i})\prod_{j\neq i} \pi_{j,r_j}\Big)\,,
\end{aligned}
\end{equation}
with $c_a$ the linearly non-separable congestion cost of~\eqref{eq:nonsep_costs}.

Considering mixed strategies, the strategy sets are topological spaces and the expected cost functions given in (\ref{eq:mixed}) are continuously differentiable.
Moreover in such continuous games, there may exist continuous potential functions, defined in \cite{Monderer96}: the gradient of these functions correspond to the expected costs.
This type of potential is widely considered in population games \cite{Sandholm01}, as it serves as a Lyapunov function for strategies' dynamics, or in games with non-atomic players~\cite{Cheung18}. In our particular setting, a generalization of these potential functions is needed, and defined as follows:

\vspace{-1mm}
\begin{definition}
An action-dependent continuous potential is a $\mathcal{C}^1$ function $F$ over mixed strategies such that, for all resources $a \in \mathcal{R}$, there exists a constant $\gamma_a$ verifying: 
\begin{equation} \forall i,\quad
    \frac{\partial F}{\partial \pi_{i,a}} (\bm{\pi}) = \gamma_a \overline{c_{i,a}}(\bm{\pi})\,.
    \label{eq:pot_continu}
\end{equation}
\end{definition}
\noindent
Note that as expected cost are $\mathcal{C}^1$ functions, such potentials $F$ are then $\mathcal{C}^2$ functions.
Continuous potential functions verify~\eqref{eq:pot_continu} with $\gamma_a=1$ for all resources $a$.
Our atomic game $\mathcal{G}$ considering continuous strategy sets of mixed strategies has an action-dependent continuous potential function. In fact, this function is the conditional expectation of the ordinal potential function $P$ when players choose pure strategies according to the mixed strategy vector $\bm{\pi}$.


\vspace{-1mm}
\begin{prop}
Atomic games $\mathcal{G}$ with linearly non-separable congestion costs have the following action-dependant continuous potential function (associated to $\gamma_a = \frac{1}{\alpha_a}, \forall a$):
\vspace{-1mm}
\begin{equation}
    F(\bm{\pi}) \triangleq \mathbb{E}_{\bm{\pi}}\left[P\right]\,,
    \label{def:adpotential}
\end{equation}
with $P$ the ordinal potential of $\mathcal{G}$.

\label{cor:lambda_pot_continu}
\end{prop}
\vspace{-1mm}
\begin{proof}
By linearity of the expected value $\mathbb{E}_{\bm{\pi}}\left[P\right]$:
\begin{equation*}
\vspace{-1mm}
\begin{aligned}
    F(\bm{\pi})&=\sum_{i,a} \pi_{i,a} \mathbb{E}_{\bm{\pi}}\left[P~|~\bm{\pi}_i = \bm{e}_a\right]\\
\vspace{-3mm}
    &= \sum_{i,a} \pi_{i,a} \frac{1}{\alpha_a} \mathbb{E}_{\bm{\pi}}\big[c_a ~|~ \pi_{i}=\bm{e}_a\big]\,,
\end{aligned}
\end{equation*}
using $c_a(\cdot) = \alpha_a \lambda(L(\cdot)) = \alpha_a P(\cdot)$.
Then, ~\eqref{eq:pot_continu} is found by differentiating by $\pi_{i,a}$, with $\gamma_a = \frac{1}{\alpha_a}$ ($\forall i,a$).
\vspace{-2mm}
\vspace{2mm}
\end{proof}


The previous proposition generalizes the particular case studied in \cite{SASTRY94} for games with common payoff, while a similar result is obtained in \cite{Bournez09} for continuous potential games.
The following proposition gives a more precise result and shows that only games with particular cost functions admit an action-dependent continuous potential function.

\begin{prop}
Finite games that admit a $\mathcal{C}^2$ action-dependant continuous potential function correspond to finite games with cost functions defined as:
\begin{equation}
    \forall i,a,\quad
    c_{i,a}(\bm{r}) \triangleq \beta_a \mu(\bm{r})\,,
    \label{eq:common}
\end{equation}
with $\beta_a$ any constant which depends on the action $a$, and $\mu$ any function of pure strategies (not necessarily increasing or linearly non-separable). 
\end{prop}
\begin{proof}
Let $F$ be a $\mathcal{C}^2$ perturbed continuous potential function, associated to constants $\gamma_a$ ($\forall a$).
Then, by Definition~\eqref{eq:pot_continu} of $F$ and according to Clairaut-Schwarz theorem (symmetry of second derivatives):
\begin{equation*}
    \forall i,j,a,b, \quad \gamma_a\frac{\partial \overline{c_{i,a}}}{\partial \pi_{j,b}} =  \gamma_b\frac{\partial \overline{c_{j,b}}}{\partial \pi_{i,a}}
     \,,
\end{equation*}
which, using~\eqref{eq:mixed}, leads to (for all mixed strategies $\bm{\pi}$):
\begin{equation*}
    \sum_{\bm{r}_{-ij}} 
    \Big(\prod_{k\neq i,j}\pi_{k,r_k} \big[\gamma_a c_{i,a}-\gamma_b c_{j,b}\big](a,b,\bm{r}_{-ij})\Big)
    =0\,.
\end{equation*}
For all pure strategies $\bm{r}_{-ij}\in\mathcal{R}^{N-2}$, last equation considered with $\bm{\pi}_k = \bm{e}_{r_k}$ ($\forall k\neq i,j$) becomes:
\begin{equation*}
   \gamma_a c_{i,a}(a,b,\bm{r}_{-ij})=\gamma_b c_{j,b}(a,b,\bm{r}_{-ij})\,.
\end{equation*}
Let $\mu = c_{i,a}$ and $\beta_b = \frac{\gamma_a}{\gamma_b}$ ($\forall b$).
Then,~\eqref{eq:common} is true.

Inversely, suppose a game with cost functions verifying~\eqref{eq:common}.
Then, $\overline{c_{i,a}}=\beta_a \mathbb{E}_{\bm{\pi}}\left[\mu ~|~ \pi_{i}=\bm{e}_a\right]$, as seen in~\eqref{eq:mixed}.
Let $F(\bm{\pi}) = \mathbb{E}_{\bm{\pi}}\left[\mu\right]$.
By linearity of the expected value, $F=\sum_{i,a} \pi_{i,a}. \mathbb{E}_{\bm{\pi}}\left[\mu~|~\bm{\pi}_i = \bm{e}_a\right]$.
Therefore,~\eqref{eq:pot_continu} is verified, with $\gamma_a = \beta_a$ ($\forall a$).
\end{proof}
\noindent
This type of games is a generalization of common payoff games with action-dependent cost.

The property of having an action-dependent continuous potential leads to convergence of simple RLA for which only local information is accessible for each player (basically her own perceived cost) in order to update her mixed strategy vector and find her best action. In fact, in most cases, players are not even aware that they are involved in a game with other players and interact through their actions. That is why the framework of online learning in which players make repeated decisions with a priori unknown rules and outcomes is suitable. In the next section, a simple RLA is described, whose convergence has already been proven when the game has a continuous potential function. We prove that there is still convergence in a case of action-dependent continuous potential and non-separable congestion game as in our setting.

\section{Reinforcement learning algorithm}
\label{RLA}

In our framework, players possess incomplete information: their only knowledge is the observation of their cost after taking an action.
\tcr{Note that best response algorithms can also be applied to game $\mathcal{G}$, but players require additional information (the exact formulation of their own cost function).}
Here, game $\mathcal{G}$ will be repeated so that players learn what their best strategy is.
More precisely, every iteration $n$ is split into two phases.
In the first phase, each player $i$ chooses an action $r_i^{(n)}$ in accordance with her mixed strategy vector $\bm{\pi}_i^{(n)}$. Thus, a vector of actions $\bm{r}^{(n)}$ is induced by the decisions of all players, which in turn implies a cost for each player based on the cost functions defined in (\ref{eq:nonsep_costs}).
Then, in the second phase of the iteration, each player updates her strategy probability vector based on the unnoisy cost. This update mechanism is a reinforcement mechanism. This type of RLA, used in stochastic games, is a linear reward-inaction scheme~\cite{SASTRY94}. Each player $i$ updates her mixed strategy vector $\bm{\pi}_i$ as follows (for any iteration $n$):
\begin{equation}
\bm{\pi_i}^{(n+1)} = \bm{\pi_i}^{(n)} + \delta \times\left(1- \frac{c_{i,a}(\bm{r}^{(n)})}{c_{\max}}\right)\times\left(\bm{e}_a-\bm{\pi_i}^{(n)}\right)\,,
\label{eq:reinforcement}
\end{equation}
with:
\begin{itemize}
\item $0<\delta <1$ the learning parameter, fixed;
\item $a=r_i^{(n)}$ the action taken by player $i$ at iteration $n$;
\item $c_\text{max} \triangleq \max_{i,a,\bm{r}} c_{i,a}(\bm{r})$.
\end{itemize}
\vspace{1mm}
The basic idea of the updating rule expressed by equation~\eqref{eq:reinforcement} is to ensure that actions prompting small or high costs are promoted or not.
This update scheme is decentralized, and the global algorithm is fully distributed. This is an important property in order to deploy it in large scale complex systems, typically congestion games with a large number of players.
The global algorithm works as follows:
\vspace{3mm}

\begin{algorithm}[h]
\KwIn{$\bm{\pi}^{(0)}$, $n=0$}
    \While{not converged}
    {
    Actions $\bm{r}^{(n)}$ according to mixed strategies $\bm{\pi}^{(n)}$\;
    \For{all players $i$}
    {
    Cost $c_{i,a}\big(\bm{r}^{(n)}\big)$ at $a=r_i^{(n)}$ given by~\eqref{eq:nonsep_costs}\;
    Update mixed strategy of $i$ with~\eqref{eq:reinforcement}\;
    }
$n \leftarrow n+1$\;
    }
\caption{RLA with synchronous global updates}
\label{algo:Sastry}
\end{algorithm}
\vspace{3mm}
This algorithm converges for games having an action-dependent continuous potential function.
The updating mechanism~\eqref{eq:reinforcement} can be written as:
\begin{equation}
    \bm{\pi}^{(n+1)}=\bm{\pi}^{(n)} + \delta G\left(\bm{\pi}^{(n)}, \bm{r}^{(n)}\right)\,,
    \label{eq:sequence}
\end{equation}
with $G$ the updating function. Let us define function  $f(\bm{\pi})\triangleq\mathbb{E}_{\bm{\pi}}[G]$ and for any $\delta$ the function $\Pi_\delta : \mathbb{R}_+\rightarrow \mathbb{R}$ as the piecewise-constant interpolation of sequence $\left(\bm{\pi}^{(n)}\right)_n$ of mixed strategies of all players. A direct application of Theorem 3.1 in~\cite{SASTRY94} demonstrates that our learning algorithm~\ref{algo:Sastry} converges weakly, as $\delta$ tends to $0$, to the solution $\Pi$ of the following Ordinary Differential Equation (ODE):
\begin{equation}
    \frac{dX}{dt} = f(X)\,,\quad X(0) = \bm{\pi}^{(0)}\,.
    \label{eq:ODE}
\end{equation}
Considering $\Pi=\left(\pi_{i,a}\right)$ and $f(\Pi)=\left(f_{i,a}\right)$, the ODE~\eqref{eq:ODE} can be written element-wise as:
\begin{equation*}
\begin{aligned}
     \forall i,a,\quad \frac{d \pi_{i,a}}{d t} &= f_{i,a} = &&\pi_{i,a} (1-\pi_{i,a}) (1-\overline{c_{i,a}}/c_{\max})\\
     & &&+\sum_{b\neq a}\pi_{i,b}(-\pi_{i,a}) (1-\overline{c_{i,b}}/c_{\max})\\
     &= -\pi_{i,a} &&\sum_{b\neq a}\pi_b\left(\overline{c_{i,a}} - \overline{c_{i,b}}\right)/c_{\max}\,.
\end{aligned}
\end{equation*}
Thus, the convergence points of our algorithm are related to the solutions of a particular ODE, which must be characterized in order to prove the convergence of the RLA to NE of game $\mathcal{G}$. In fact, next proposition proves that having an action-dependent continuous potential implies convergence to pure NE of our RLA described in Algorithm~\ref{algo:Sastry}.
\begin{prop}
If a finite game $\mathcal{G}$ has an action-dependent continuous potential then, for any initial non-pure strategies $\bm{\pi}^{(0)}$, function $\Pi\triangleq\lim_{\delta\rightarrow 0}\Pi_\delta$ converges to a pure NE.

\label{prop:reinforcement}
\end{prop}
\begin{proof}
This proof is inspired by the one of Theorem 3.3 of~\cite{SASTRY94}.
Here, the continuous potential $F$ of game $\mathcal{G}$ is action-dependant and associated with constants $\gamma_a$:
\begin{equation*}
\begin{aligned}
    \frac{d F}{d t}\left(\Pi\right) &= 
    -\sum_{i\in\mathcal{N}}\sum_{k=1}^M\sum_{l>k}\gamma_{a_k}  \pi_{i, a_k}\pi_{i,a_l} \frac{\left(\overline{c_{i,a_k}}-\overline{c_{i,a_l}}\right)^2}{c_{\max}}<0 \,, 
\end{aligned}
\end{equation*}
with $\mathcal{R}=\{a_1,\dots,a_M\}$ the set of resources.
Then, $t\mapsto F(\Pi(t))$ is non-increasing.
Therefore, the RLA always converges to a pure NE.
\end{proof}

Note that Algorithm~\ref{algo:Sastry} works with synchronous global updates, meaning that all players update their strategy simultaneously at each iteration.
An asynchronous version of Algorithm~\ref{algo:Sastry} can be considered, where only one player updates her strategy at each iteration.
A player $i$ is chosen for iteration $n$ with uniform probability $p_i\triangleq\frac{1}{N}$. The convergence of such an asynchronous algorithm towards a pure $\varepsilon$-NE\footnote{For any $\varepsilon\geq 0$, a  pure strategy vector $\bm{r}^*$ is an $\varepsilon$-Nash Equilibrium if $\forall i,a,\quad
    c_{r_i^*}(\bm{r}^*) \leq c_{a}\left(a, \bm{r}_{-i}^*\right) + \varepsilon$.}
has been proven in~\cite{Bournez09} for games having an ordinal potential function. This result can be applied directly to our framework of finite congestion games with linearly non-separable costs, as an ordinal potential function exists (see Prop.~\ref{prop:pot}). In next Section, we illustrate the convergence of both synchronous and asynchronous RLA in a finite smart charging congestion game with linearly non-separable costs.

\section{A smart charging game}
\label{Smart}

\subsection{Description}

\begin{figure}
    \centering
    \includegraphics[width = 0.32\textwidth]{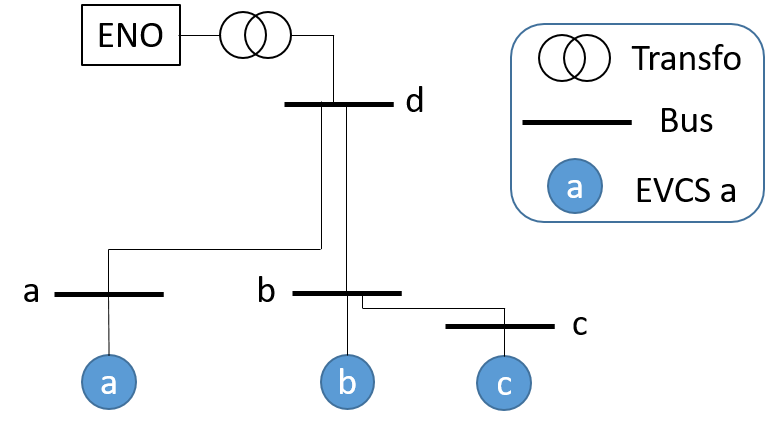}
    \caption{\small Electrical network (with ENO its operator) containing the three EVCS.}
    \label{fig:schema}
\vspace{-6mm}
\end{figure}

In this section, the synchronous reinforcement learning algorithm and its asynchronous version are illustrated on a real-life example of a finite congestion game with linearly non-separable costs.
The players are Electric Vehicle (EV) users who choose at which Charging Station (EVCS) they charge their EV battery (here, the EVCS are the resources).
These EVCS are part of an electrical grid network shown in Fig.~\ref{fig:schema}.
\tcr{For simplicity, both the EVCS and the grid are supposed to be managed by the same Electric Network Operator (ENO).}
Note that the learning algorithm of previous section can be applied to more general grid topologies.
At each electric bus $r$, there is some (fixed) power demand $L_r^0$ corresponding to other usages than EV charging.
Choices of EV users only depend on the \tcr{observed} cost $c_r$ of their charging operation at EVCS $r$, proportional to the charging unit price $p_r$.
All users are supposed to have the same charging need $\rho$ (in kWh) to simplify notations, but this assumption (symmetric game) can be readily relaxed, in line with Remark~\ref{rmq:hetero}. Note that the charging operation is supposed to last exactly one hour, at constant power (then equal to $\rho$, in kW). Thus, the charging cost for any EV user at EVCS $r$ is $c_r\triangleq\rho \times p_r$.

In our framework, the ENO chooses charging unit prices $p_r$ as incentives to reduce its global cost $H$. 
Here, marginal cost pricing schemes are considered: $p_r \triangleq \partial H / \partial L_r$ with $L_r \triangleq L^0_r + n_r(\bm{r}) \rho$ the total power demand at EVCS $r$. The ENO global cost $H$ is defined as the cost of the electricity generation needed to fulfill the total power need $\bm{\mathcal{L}}\triangleq(L_r)_r$ in the grid network considered. Typically, this cost is a quadratic function~\cite{Mohsenian10} of the (apparent) power $S_0$ at the head of the grid:
\begin{equation}
    H(\bm{\mathcal{L}}) \triangleq \eta S_0^2(\bm{\mathcal{L}})\,,
\end{equation}
with $\eta$ a scaling constant. This power can be obtained by running a Power Flow algorithm, more precisely the Bus Injection Model~\cite{Hu15}. This algorithm solves the power balance equation at each bus (between the given power production/load $S_{0,k}$ at bus $k$ and power flows $S_k$ from/to the bus):
\vspace{-3mm}
\begin{equation}
S_{0,k} = U_k \sum_{m\in X_k} \overline{Y_{k,m}}\overline{U_m} ~(\triangleq S_k)\,,
\label{eq:pf}
\end{equation}
with $U_k$ the complex voltage at bus $k$, $X_k$ the set of buses connected to bus $k$ and $Y_{k,m}$ the admittance of the line between buses $k$ and $m$.

\tcr{The charging unit price $p_r = \partial H / \partial L_r (\bm{\mathcal{L}})$ at EVCS $r$ is not necessarily a function of a linear combination of the action vector $\bm{r}$ as in~\eqref{eq:nonsep_costs}, due to the complexity of power flow equations~\eqref{eq:pf}.
Therefore, the grid is reduced to only the transformer bus $d$, as if all the power demand came from it, and the grid topology is replaced by constants $\bm{\Tilde{\alpha}}\triangleq(\Tilde{\alpha}_r)$:} 
\begin{equation}
\vspace{-1mm}
\frac{\partial H}{\partial L_r}(\bm{\mathcal{L}})
\simeq \frac{\partial H}{\partial L_r}(\bm{\Tilde{\mathcal{L}}})
\Big(= \Tilde{\alpha}_r \frac{\partial H}{\partial L_d}(\bm{\Tilde{\mathcal{L}}})\Big)\,, ~ \text{ with}
\end{equation}
\vspace{-2mm}
\begin{equation*}
\bm{\mathcal{L}} \triangleq (L_a\,, L_b\,, L_c\,, 0) ~\text{ and }~ \bm{\Tilde{\mathcal{L}}} \triangleq (0\,, 0\,, 0\,, \sum_r \Tilde{\alpha}_r L_r)\,.
\vspace{-2mm}
\end{equation*}
The meaning of constant $\Tilde{\alpha}_r$ is that a power demand $L_r$ at EVCS $r$ is equivalent (i.e., yields equal marginal costs) to a power demand $\Tilde{\alpha}_r L_r$ at the transformer bus $d$.
Coefficients $\bm{\Tilde{\alpha}}$ are obtained by solving the following square minimization:
\vspace{-2mm}
\begin{equation}
    \bm{\Tilde{\alpha}} \triangleq \argmin_{\bm{\alpha}>0} \int_0^{L^\text{max}} \Big[\frac{\partial H}{\partial L_a}
    \left(L^0_a + L_a\,, L^0_b\,,L^0_c\,,0\right)
    \label{eq:min}
\end{equation}
\begin{equation*}
    - \alpha_a\frac{\partial H}{\partial L_d}
\Big(\sum_r\alpha_r L^0_r + \alpha_a L_a\Big)
    \Big]^2 \text{d}L_a + \int (b) + \int(c) \,,
\end{equation*}
where $\int (b) + \int(c)$ means that the first integral in~\eqref{eq:min} is repeated for EVCS $b$ and $c$.


Then, charging unit prices $p_r\triangleq \Tilde{\alpha}_r \partial H / \partial L_d (\bm{\Tilde{\mathcal{L}}})$ lead to the following charging cost $c_r$ for EV users choosing EVCS $r$:
\vspace{-3mm}
\begin{equation}
        c_r = \rho \times \Tilde{\alpha}_r \frac{\partial H}{\partial L_d}\left(\sum_s \rho\Tilde{\alpha}_s n_s(\bm{r}) + \sum_s\Tilde{\alpha}_s L_s^0\right)\,.
\end{equation}
This cost function verifies Definition~\ref{def:nonsep} with $\alpha_r \triangleq \rho \Tilde{\alpha}_r$ and $\lambda \triangleq \partial H / \partial L_d$, an increasing function of $L(\bm{r})=\sum_s \rho\Tilde{\alpha}_s n_s(\bm{r})$.

\subsection{Numerical illustrations}
The parameters of the smart charging game are set as follows:
a typical French 20kV electrical network has around 1500 customers, with a standard 6kVA contract power.
The total number of EV users is set to $N=1500$ and the energy need is $\rho = 3$kWh, half of the daily mean individual EV consumption in France\footnote{Enquete Nationale Transports et Deplacements:
\url{https://utp.fr/system/files/Publications/UTP_NoteInfo1103_Enseignements_ENTD2008.pdf} (in French).}.
Based on bus lines characteristics and the grid network topology considered in Fig.~\ref{fig:schema}, for each EVCS $r$, linear coefficients $\Tilde{\alpha}_r$ are obtained solving the square minimization problem (\ref{eq:min}) :
$\Tilde{\alpha}_a = 1.12$, $\Tilde{\alpha}_b = 1.07$ and $\Tilde{\alpha}_c=1.18$.
Note that EVCS $a$ and $c$ have a greater impact on ENO global cost $H$ than EVCS $b$ because they are further away from the transformer (see Fig.~\ref{fig:schema}).
The ENO global cost $H$ is adjusted with $\eta = 5. 10^{-3}$ (\euro/MVA$^2$) so that charging unit prices $p_r$ remain between 10 and 20c\euro/kWh. Regarding learning characteristics, the learning parameter is set to $\delta = 0.5$ and the initial mixed strategies are equally distributed among resources: $\pi_{i,r} = 1/3$ for all users $i$ and EVCS $r=a,b,c$. Similarly, for the asynchronous version of Algorithm~\ref{algo:Sastry}, each player $i$ is chosen with a probability $p_i=1/N$ for the update phase.

\begin{figure}
    \centering
    \includegraphics[width = 0.5\textwidth]{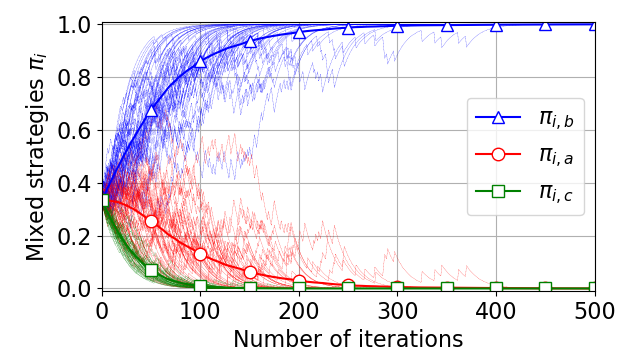}
    \caption{\small Evolution of mixed strategies $\bm{\pi}_i$ (of only 50 users) throughout iterations of Algorithm~\ref{algo:Sastry}.
    Thicker lines represent the average (over all users) mixed strategy.}
    \label{fig:sastry}
\vspace{-5mm}
\end{figure}

Figures~\ref{fig:sastry} and~\ref{fig:bournez} show the evolution of mixed strategies $\bm{\pi}_i$ (of only 50 users, for visibility) throughout iterations respectively for the synchronous Algorithm~\ref{algo:Sastry} and for the asynchronous version.
As mentioned in Section~\ref{sec:pure}, although the pure NE of this game (all users choosing EVCS with lowest impact on grid costs) may seem trivial, EV users need hundreds of iterations to learn it (see Fig.\ref{fig:sastry}), as they have no information on the grid topology.
Considering thicker lines (average mixed strategy over all users), it can be seen that, while the synchronous algorithm converges in less than 500 iterations (Fig.~\ref{fig:sastry}), it takes more than thousand times as many iterations for the asynchronous one (Fig.~\ref{fig:bournez}).
This is understandable since in the asynchronous version, only one player updates her strategy at each iteration, instead of all players like in the original Algorithm~\ref{algo:Sastry}.
This also explains why larger number of players lead to slower convergence for the asynchronous version, while it has no effect on the original Algorithm~\ref{algo:Sastry}.
Note that the evolution of mixed strategies may slightly vary from one execution to another (of either algorithm), due to actions randomly chosen from mixed strategies.
Similarly, the number of iterations until convergence depends on the initial mixed strategies $\bm{\pi}^{(0)}$ and decreases with the learning parameter $\delta$.

\vspace{-1mm}
\section{Conclusions and perspectives}
\vspace{-1mm}
\label{conc}
In this paper, a reinforcement learning algorithm (RLA) has been applied to obtain in a fully decentralized manner a pure Nash equilibrium (NE) for a finite congestion game with non-separable cost functions.
Non-separability of cost functions makes it difficult to prove the existence of pure NE and the convergence property of RLA in such particular congestion games.
When cost functions are separable, all these results come from the existence of an exact potential. However, assuming costs are linearly non-separable, we are able to prove the existence of an ordinal potential function, which can serve as a Lyapunov function for proving the convergence of simple RLA. Moreover, these results were applied to a smart charging game in which EV users choose selfishly a charging station depending on the smart charging price, which is based on a Power flow solution of the impact of EV charging on the grid.
\tcr{In a future work, a rationality coefficient for players will be added to the RLA.}

\begin{figure}
    \centering
    \includegraphics[width = 0.5\textwidth]{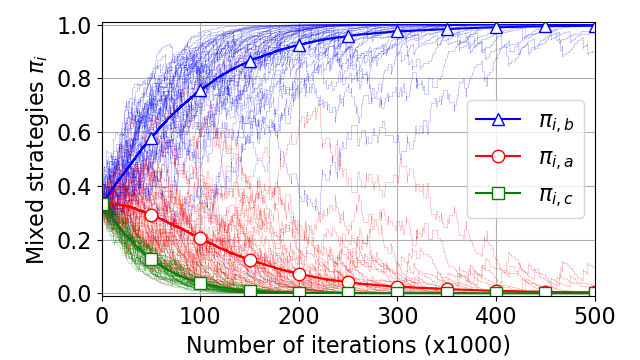}
    \caption{\small Evolution of mixed strategies $\bm{\pi}_i$ (of only 50 users) throughout iterations of asynchronous version of Algorithm~\ref{algo:Sastry}}
    \label{fig:bournez}
    \vspace{-3mm}
\end{figure}    

\end{document}